\documentclass[sigconf, screen]{acmart}
\AtBeginDocument{%
  }

\usepackage{pifont}

\usepackage{hyperref}
\usepackage{url}
\usepackage{bm}
\usepackage{enumitem}
\usepackage{wrapfig}
\usepackage{booktabs}
\usepackage{multirow}
\usepackage{caption}
\usepackage{graphicx}
\usepackage{subfigure}
\usepackage{wrapfig}
\usepackage{mathtools}
\usepackage{amsmath,amsfonts, amsthm}
\usepackage{tikz}
\usepackage{algorithm}
\usepackage{algpseudocode}
\usepackage{threeparttable}
\usepackage{xspace}
\usepackage{tcolorbox}

\newcommand{\myvector}{\includegraphics[height=1em]{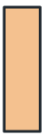}\xspace}

\PassOptionsToPackage{dvipsnames}{xcolor,colortbl}
\RequirePackage{xcolor,colortbl}[dvipsnames] 
\definecolor{darkcerulean}{rgb}{0.03, 0.27, 0.49}
\definecolor{burgundy}{rgb}{0.7, 0.0, 0.13}
\definecolor{majorelleblue}{rgb}{0.38, 0.31, 0.86}
\definecolor{forestgreen}{rgb}{0.13, 0.55, 0.13}
\definecolor{arylideyellow}{rgb}{0.89, 0.61, 0.06}
\hypersetup{
    colorlinks=true, 
    linkcolor=burgundy,    
    urlcolor=majorelleblue     
}
\usepackage{array}
\usepackage{arydshln}
\algnewcommand{\LineComment}[1]{\State\textcolor{blue}{\(\triangleright\)#1}}

\setcopyright{acmlicensed}
\copyrightyear{2018}
\acmYear{2018}
\acmDOI{XXXXXXX.XXXXXXX}
\acmConference[Conference acronym 'XX]{Make sure to enter the correct
  conference title from your rights confirmation emai}{June 03--05,
  2018}{Woodstock, NY}
\acmISBN{978-1-4503-XXXX-X/18/06}

\begin{document}

\title{CG-RAG: Research Question Answering by Citation Graph Retrieval-Augmented LLMs}

\author{Yuntong Hu}
\email{yuntong.hu@emory.edu}
\affiliation{%
  \institution{Emory University}
  \city{Atlanta}
  \state{GA}
  \country{USA}
}

\author{Zhihan Lei}
\email{zhihan.lei@emory.edu}
\affiliation{%
  \institution{Emory University}
  \city{Atlanta}
  \state{GA}
  \country{USA}
}

\author{Zhongjie Dai}
\email{dzj@tongji.edu.cn}
\affiliation{%
  \institution{Tongji University}
  \city{Shanghai}
  \country{China}
}

\author{Allen Zhang}
\email{azhang490@gatech.edu}
\affiliation{%
  \institution{Georgia Institute of Technology}
  \city{Atlanta}
  \state{GA}
  \country{USA}
}

\author{Abhinav Angirekula}
\email{	aa125@illinois.edu}
\affiliation{%
  \institution{University of Illinois Urbana-Champaign}
  \city{Urbana}
  \state{IL}
  \country{USA}
}

\author{Zheng Zhang}
\email{zheng.zhang@emory.edu}
\affiliation{%
  \institution{Emory University}
  \city{Atlanta}
  \state{GA}
  \country{USA}
}

\author{Liang Zhao}
\email{liang.zhao@emory.edu}
\affiliation{%
  \institution{Emory University}
  \city{Atlanta}
  \state{GA}
  \country{USA}
}

\renewcommand{\shortauthors}{Yuntong et al.}

\begin{abstract}
Research question answering requires accurate retrieval and contextual understanding of scientific literature. However, current Retrieval-Augmented Generation (RAG) methods often struggle to balance complex document relationships with precise information retrieval. 
In this paper, we introduce Contextualized Graph Retrieval-Augmented Generation (CG-RAG), a novel framework that integrates sparse and dense retrieval signals within graph structures to enhance retrieval efficiency and subsequently improve generation quality for research question answering. First, we propose a contextual graph representation for citation graphs, effectively capturing both explicit and implicit connections within and across documents. Next, we introduce Lexical-Semantic Graph Retrieval (LeSeGR), which seamlessly integrates sparse and dense retrieval signals with graph encoding. It bridges the gap between lexical precision and semantic understanding in citation graph retrieval, demonstrating generalizability to existing graph retrieval and hybrid retrieval methods. Finally, we present a context-aware generation strategy that utilizes the retrieved graph-structured information to generate precise and contextually enriched responses using large language models (LLMs). Extensive experiments on research question answering benchmarks across multiple domains demonstrate that our CG-RAG framework significantly outperforms RAG methods combined with various state-of-the-art retrieval approaches, delivering superior retrieval accuracy and generation quality.
\end{abstract}

\ccsdesc[500]{Computing methodologies~Natural language generation}
\ccsdesc[500]{Information systems~Retrieval models and ranking}

\keywords{Graph retrieval-augmented generation, research question answering, hybrid retrieval, graph learning, citation graphs}

\received{20 February 2007}
\received[revised]{12 March 2009}
\received[accepted]{5 June 2009}

\maketitle

\section{Introduction}

\begin{figure*}[tb]
  \centering
  \includegraphics[width=\textwidth]{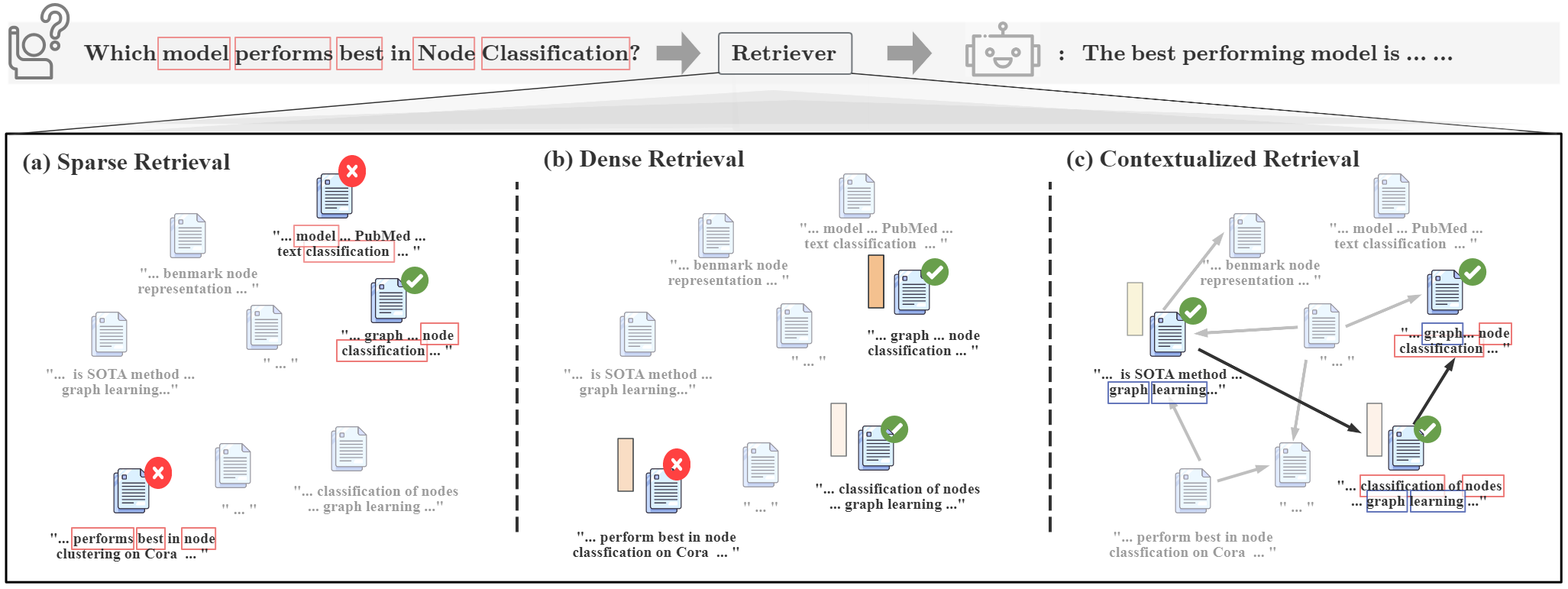}
  \caption{Illustration of retrieval-augmented research question answering using: (a) sparse retrieval based on lexical matches, (b) dense retrieval based on semantic relevance, and (c) contextualized retrieval leveraging graph context, i.e., interactions between documents. \myvector represents the dense embedding, where deeper colors indicate a higher semantic relevance to the question. \textcolor{red}{\fbox{\textcolor{black}{Boxed text}}} in red highlights the matched terms between questions and documents, while \textcolor{blue}{\fbox{\textcolor{black}{Boxed text}}} in blue highlights the matched terms between documents.}
  \vspace{-10pt}
  \label{fig:intro}
\end{figure*}

Question answering is a critical domain recently driven by advancements in Large Language Models (LLMs). While LLMs exhibit exceptional capabilities in addressing commonsense questions \citep{lehmann2024large, hu2023beyond, chang2024survey}, their pre-trained knowledge inevitably becomes outdated over time, rendering them insufficient for delivering timely, precise, and comprehensive answers, particularly for complex scientific and domain-specific questions. Additionally, the substantial cost of continuously fine-tuning and updating LLMs makes maintaining their relevance impractical.

For open-domain question answering, which requires relevant contexts for precise answers, Retrieval-Augmented Generation (RAG) \citep{lewis2020retrieval} offers a promising solution. By integrating external knowledge, RAG addresses the limitations of static pre-trained LLMs, enabling access to up-to-date, domain-specific information and thereby enhancing answer accuracy \citep{siriwardhana2023improving, taffa2023leveraging, wang2024rear}.
RAG comprises two components: a retriever, which identifies relevant information from the database based on the query, and a generator, which utilizes the retrieved information to construct responses. The quality of generation heavily depends on the effectiveness of the retrieval process \citep{yu2024evaluation, zhao2024retrieval}. 
Traditional retrieval methods primarily focus on lexical and semantic relevance to evaluate the relevance between documents and queries, utilizing sparse and dense retrieval signals, respectively \citep{luan2021sparse, zhao2024dense}. For graph-based databases, however, these methods fall short, as they fail to account for intricate inter-document relationships. For example, simple retrieval methods often overlook critical citation links in paper citation networks, which are essential for capturing nuanced connections between documents.

Unlike typical QA tasks, which range from layman-level to domain-expert-level based on well-curated documents or knowledge graphs, this paper focuses on more challenging QA at the research frontier, which can only be addressed by analyzing the body of research papers, a task referred to as \textbf{research question answering}. Research question answering necessitates considering connections between documents through citation links to ensure comprehensiveness \citep{lozano2023clinfo, giglou2024scholarly, wang2024biorag}. In citation graphs, the relevant papers are connected not only semantically but also via citation links, encapsulating structured contextual information. Leveraging this contextual information is essential for enhancing both retrieval accuracy and the quality of generated responses, as illustrated in \hyperref[fig:intro]{Figure} \ref{fig:intro}. Specifically, as shown in \hyperref[fig:intro]{Figure} \ref{fig:intro}\hyperref[fig:intro]{(c)}, each paper requires an evaluation of lexical and semantic relevance, as well as its citations to other relevant documents. 
For instance, a paper focused on "graph representation learning" is theoretically correlated to a query on "node classification," but it cannot be retrieved using sparse or dense retrieval alone. When multi-hop citation links are considered, however, it becomes highly relevant. Notably, not all citation-linked papers are pertinent to the query, requiring an approach that integrates lexical and semantic relevance, structural patterns, and graph-based context to ensure accurate retrieval.

To tackle this problem, we propose \textbf{\textit{Contextualized Graph Retrieval-Augmented Generation (CG-RAG)}}, a novel framework tailored for research question answering. An important consideration of research question answering is that research paper content is highly heterogeneous: for instance, the information in the related work, methodology, and experimental sections each serve different purposes and answer distinct types of questions. Therefore, effective modeling requires breaking down documents into semantic chunks—coherent sections representing specific aspects of a research paper. To achieve this, we mathematically construct a citation graph at the chunk level, where each chunk represents a semantically meaningful module of the paper. These chunks form a graph structure with intra-paper and inter-paper connections, capturing both internal coherence and external relationships. To represent this, we introduce the \hyperref[sec:graph]{\textit{\textbf{Contextual Citation Graph}}}, which decomposes citation graphs into chunk-level granularity, enabling fine-grained relationship discovery.

To synergize lexical and semantic retrieval signals inside and across networked documents in citation graphs, we propose the \hyperref[sec:fusion]{\textit{\textbf{\underline{Le}xical-\underline{Se}mantic \underline{G}raph \underline{R}etrieval (LeSeGR)}}} method. This approach convolves over query-relevant subgraphs, where edges and nodes are characterized by lexical and semantic relevance scores. Importantly, we theoretically demonstrate that the existing post-retrieval paradigm is a special case of our approach. When documents are contextually linked, entangling retrieval signals through graph structures enhances retrieval performance.
Finally, the subgraph embeddings are used as input to LLMs for generating final answers. 
Our experiments conducted on citation graphs from diverse scientific domains demonstrate the superior retrieval and generation effectiveness of CG-RAG. Furthermore, our evaluations reveal that retrieval-augmented LLMs equipped with CG-RAG outperform state-of-the-art retrieval strategies, significantly enhancing the quality of LLM-generated responses.

The rest of the paper is organized as follows. \hyperref[sec:related work]{Section} \ref{sec:related work} reviews related work on retrieval-augmented generation. \hyperref[sec:problem]{Section} \ref{sec:problem} highlights key aspects of retrieval-augmented research question answering within the context of citation graphs. \hyperref[sec:method]{Section} \ref{sec:method} introduces a novel formulation that extends beyond the current retrieval paradigm and details our proposed retrieval strategy and retrieval-augmented generation method. \hyperref[sec:exp]{Section} \ref{sec:exp} presents experimental results, comparing our approach with RAG frameworks using various state-of-the-art retrieval methods for research question answering. Finally, \hyperref[sec:conclution]{Section} \ref{sec:conclution} concludes the paper with key insights.

\section{Related Work}\label{sec:related work}
\subsection{Information Retrieval}
Information retrieval (IR) focuses on extracting relevant information from large corpora.
Two primary retrieval techniques dominate the field: \textit{sparse retrieval} and \textit{dense retrieval}. Sparse retrieval methods, such as TF-IDF \citep{salton1975vector} and BM25 \citep{robertson2009probabilistic}, rely on term-based representations to evaluate lexical matches between queries and documents. These approaches perform well in scenarios where exact term matching is essential, but they struggle with semantic meaning.
In contrast, dense retrieval methods leverage pre-trained language models such as BERT \citep{devlin2018bert} to encode queries and documents as continuous, low-dimensional embeddings, capturing semantic similarity through maximum inner product search (MIPS) \citep{reimers2019sentence, xiao2024c, karpukhin2020dense, xiong2020approximate}. Dense retrieval effectively overcomes the lexical gap, retrieving semantically related results even when query terms differ from the document's terminology.
Recently, \textit{hybrid retrieval} techniques have also emerged to combine the strengths of sparse and dense methods while addressing their respective limitations \citep{luan2021sparse, mandikal2024sparse, novotny2022combining}. By integrating sparse and dense signals, these approaches provide a robust solution for retrieving relevant information from long and complex documents.

\subsection{Retrieval-Augmented Generation} 
Retrieval-Augmented Generation (RAG) is a technique that integrates external retrieval systems to enhance large language models (LLMs) \citep{lewis2020retrieval, li2022survey, gao2023retrieval}. Unlike traditional LLMs that rely solely on pre-trained knowledge, RAG leverages external sources during inference, enabling more accurate and up-to-date responses. This makes RAG particularly effective for specialized tasks, such as literature-based question answering \citep{zakka2024almanac, muludi2024retrieval}.
Naturally, current literature question-answering systems are predominantly built upon RAG frameworks \citep{wang2021literatureqa, stocker2023fair, auer2020improving, giglou2024scholarly}, relying mainly on dense retrieval combined with LLMs. Literature data, however, is inherently graph-structured, where topological information, such as hierarchical and citation relationships, plays a crucial role in the retrieval and generation processes. As such, most existing literature question answering methods fail to incorporate this structural context effectively.
Recently, GraphRAG \citep{hu2024grag, edge2024local, he2024g} was introduced to extend RAG to graph-related scenarios, offering the potential to capture complex interconnections in literature datasets by leveraging graph-based relationships such as hierarchical and citation structures.
In our study, we focus on leveraging graph context to enhance the retrieval and generation processes in literature question answering.

\section{Research Question Answering}\label{sec:problem}
\textbf{Citation Graph.}\quad A citation Graph, $\mathcal{G} = (V, E, \{d_v\}_{v \in V})$, consists of a node set $V$ and an edge set $E$, where each node $v \in V$ is associated with natural language attributes in a paper. These papers, represented by $\mathcal{D} = \{d_v\}_{v \in V}$, include textual information such as abstracts, sections, and other relevant content.

\vspace{2pt}
\noindent \textbf{Research Question Answering.}\quad Given a query \(q\) over the citation graph \(\mathcal{G}\), the objective of research question answering is to generate an appropriate answer by leveraging the relevant information retrieved from \(\mathcal{G}\).
Formally, this objective is defined as:
\begin{align}
    p_{\theta}(Y|q, \mathcal{G}) = \arg\max_{\theta} \prod_{i=1}^{n} p_{\theta}(y_i|y_{<i}, q, \mathcal{G}),
    \label{output}
\end{align}
where \(\theta\) represents the parameters of the generative language model, \(Y = \{y_i\}_{i=1}^{n}\) is the generated answer sequence, and \(y_{<i}\) denotes the prefix tokens of the sequence up to position \(i-1\).

The quality of the generated answer is highly dependent on the effectiveness of the retrieval process within the citation graph.
Let \( f_o \) denote the relevance scoring function for a retrieval system \( o \). Recent advances in graph-based retrieval leverage structural information, formalized as:  
\begin{equation}
f_{\text{graph}}: g(f_{\text{dense}}) \to \mathbb{R},
\end{equation}  
where \( g(\cdot) \) encodes topological information. These methods primarily rely on dense representations, which often fail to capture sparse lexical matches—essential in citation graphs for identifying exact cross-references and key terms critical to retrieval accuracy.
Thus, integrating sparse and dense retrieval is essential.  
Existing hybrid retrieval systems combine these signals via post-retrieval fusion:  
$f_{\text{hybrid}}: f_{\text{sparse}} \bigoplus f_{\text{dense}} \to \mathbb{R},$
where \( \bigoplus \) denotes operations such as score fusion. Treating documents as isolated entities, however, limits their effectiveness in structured databases such as citation graphs. Designing a retrieval system that is both lexically and semantically aware in graph retrieval remains an unresolved challenge.

\section{Methodology}\label{sec:method}
\subsection{Overview} 

To overcome these limitations, we propose Contextualized Graph Retrieval-Augmented Generation (CG-RAG), introducing a novel retrieval method called \textbf{\underline{Le}xical-\underline{Se}mantic \underline{G}raph \underline{R}etrieval (LeSeGR)}, which integrates discrete sparse signals and continuous dense signals in a manner that respects the graph topology. Formally, the paradigm is defined as:
\begin{equation}
f_{\text{entangled}}: g(f_{\text{sparse}} \bigotimes f_{\text{dense}}) \to \mathbb{R},
\end{equation}
where \(g(\cdot)\) is a graph encoder that incorporates structured context during retrieval, and \(\bigotimes\) represents the entangled fusion of sparse and dense signals. The transition to $f_{\text{entangled}}$ offers greater generality and capabilities, but requires addressing fundamental challenges:
\begin{itemize}[left=0pt]
    \item \hyperref[sec:graph]{Section} \ref{sec:graph} introduces the contextual citation graph, which captures chunk-level cross-relationships by surrounding each chunk with its relevant context.
    \item \hyperref[sec:fusion]{Section} \ref{sec:fusion} introduces the Lexical-Semantic Graph Retrieval that integrates sparse and dense signals within graph-based scenarios. We prove that it encompasses existing hybrid retrieval methods based on post-retrieval approaches as a special case when graph contextual information is absent and extends dense-signal-only graph retrieval, highlighting its generalizability to current retrieval frameworks.
    \item \hyperref[sec:generation]{Section} \ref{sec:generation} introduces Contextualized Graph Retrieval-Augmented Generation that leverages retrieved contextual subgraphs by LeSeGR to improve the quality of generated responses.
\end{itemize}

\subsection{Contextual Citation Graph}\label{sec:graph}
Research paper content is highly diverse, with sections like related work, methodology, and experiments serving distinct purposes and answering different questions, necessitating the segmentation of documents into semantic chunks that capture specific aspects. Given a citation graph \(\mathcal{G} = (V, E, \{d_v\}_{v \in V})\), each document \(d_v\) is decomposed into a set of chunks \(C_v\), with all chunks collectively forming \(\mathcal{C} = \bigcup_{v \in V} C_v\).  
Chunks may reference each other within the same document (\textit{intra-connections}) or across different documents (\textit{inter-connections}). Intra-connections are explicit, such as a section referencing earlier subsections within the same paper, while inter-connections occur when one paper cites another without explicitly linking to specific chunks within it. Intra-document links typically provide highly relevant context, whereas inter-document links offer supplementary but less significant information.
To capture these interactions during retrieval, we propose the hierarchical citation graph, modeling relationships both within and across documents. Formally, it is defined as \(\bar{\mathcal{G}} = (\bar{V}, \bar{E}, \mathcal{C})\), where \(\mathcal{C} = \{c_i\}_{i \in \bar{V}}\) represents the set of chunks.

\subsubsection{Chunk Node} 
Each chunk $c_i \in \mathcal{C}$ corresponds to a fixed-length segment of text extracted from a document in the citation graph $\mathcal{G}$. Specifically, documents are divided into chunks with a maximum token length of $l$, such that a document with a total token length of $L$ is divided into $\lceil \frac{L}{l} \rceil$ chunks. These chunks serve as the nodes in the hierarchical citation graph $\bar{\mathcal{G}}$.

\subsubsection{Chunk-Chunk Edge} 
The edges in $\bar{\mathcal{G}}$ represent relationships between chunks, capturing both intra- and inter-document connections. Edge weights reflect the strength and nature of these relationships.

\vspace{2pt}
\noindent \textbf{Intra-document Edges} connect chunks within the same document, preserving the logical flow and structural hierarchy of the text. When $c_i, c_j \in C_v$, an edge is established if a structural dependency exists between them, such as adjacency or explicit cross-references. 
Adjacency refers to the logical connection between two consecutive chunks, where $c_j$ precedes $c_i$, resulting in an edge $c_j \rightarrow c_i$. Explicit cross-references occur when $c_i$ refers to $c_j$, establishing an edge $c_j \rightarrow c_i$.

\vspace{2pt}
\noindent \textbf{Inter-document Edges}, in contrast, connect chunks across different documents. For a chunk \(c_i \in C_u\), the Top-\(n\) relevant chunks in \(C_v\) are linked to \(c_i\) if $(v,u)\in E$. The relevance \(r_{ij}\) between \(c_i\) and \(c_j\) is computed using the relevance scoring functions $f_{\text{sparse}}$ and $f_{\text{dense}}$: $f_{\text{sparse}}(\mathbf{c}_i^{\text{sparse}}, \mathbf{c}_j^{\text{sparse}}) + f_{\text{dense}}(\mathbf{c}_i^{\text{dense}}, \mathbf{c}_j^{\text{dense}})$, where $\mathbf{c}_i^{\text{sparse}}$ and $\mathbf{c}_i^{\text{dense}}$ are the sparse and dense representations of chunk $c_i$, respectively, and similarly for $c_j$. 

\subsection{Lexical-Semantic Graph Retrieval (LeSeGR)}\label{sec:fusion}
Given a hierarchical citation graph \(\bar{\mathcal{G}} = (\bar{V}, \bar{E}, \mathcal{C}, w)\), the representation of each chunk \(c_i \in \mathcal{C}\) is designed to combine the advantages of sparse lexical vectors and dense semantic vectors. Additionally, if the contexts around \(c_i\) contain relevant information, the representation incorporates contributions from these contextual chunks. This forms the basis of an entangled representation, integrating both sparse-dense fusion and graph contextual information.

\begin{algorithm}
\caption{Lexical-Semantic Graph Retrieval.}
\label{alg:retrieval}
\begin{algorithmic}[1]
\Require Citation graph \(\mathcal{G} = (V, E, \{d_v\}_{v \in V})\), Query \(q\)
\Ensure A list of contextual subgraphs \(\{\bar{\mathcal{G}}\}\)

\LineComment{Initialize graph structures and representations.}
\If{cached contextual graph \(\bar{\mathcal{G}}\) exists}
    \State Load \(\bar{\mathcal{G}} = (\bar{V}, \bar{E}, \mathcal{C})\)
\Else
    \State Generate \(\bar{\mathcal{G}} = (\bar{V}, \bar{E}, \mathcal{C})\) from \(\mathcal{G}\)
    \State Cache \(\bar{\mathcal{G}}\) for future use
\EndIf
\LineComment{Initialize sparse and dense representations for the query and all chunks.}
\State \(\mathbf{c}_i^{\text{sparse}}, \mathbf{c}_i^{\text{dense}} \gets\) sparse and dense encoders for \(\forall c_i \in \mathcal{C}\)
\State \(\mathbf{q}^{\text{sparse}}, \mathbf{q}^{\text{dense}} \gets\) sparse and dense encoders for $q$ 

\LineComment{Compute initial query relevance for all chunks.}
\For{each chunk \(c_i \in \mathcal{C}\)}
    \State \(\delta_{qi} \gets f_{\text{sparse}}\mathbf{q}^{\text{sparse}},\mathbf{c}_i^{\text{sparse}})\) \Comment{Sparse Signal}
    \State \(\alpha_{ij} \gets \text{MLP}_{\phi}(\mathbf{c}_i^{\text{dense}} \ominus \mathbf{c}_j^{\text{dense}})\) \Comment{Dense Signal}
\EndFor

\LineComment{Perform message passing through the graph.}
\For{each layer \(k = 1, \dots, K\)}
    \For{each chunk \(c_i \in \mathcal{C}\)}
        \LineComment{Compute messages from neighbors and itself.}
        \State \(\mathbf{m}_j^{(k)} \gets \text{MSG}^{(k)}\big(\delta_{qj} \cdot \alpha_{ij} \cdot \mathbf{h}_j^{(k)}\big) \forall j \in \{i\} \cup \mathcal{N}(i)\)
        \LineComment{Aggregate messages.}
        \State \(\mathbf{h}_i^{(k+1)} \gets \text{AGG}^{(k)}\big(\{\mathbf{m}_j^{(k)} \mid j \in \{i\} \cup \mathcal{N}(i)\}\big)\)
    \EndFor
\EndFor

\LineComment{Compute relevance scores with entangled representations.}
\For{each chunk \(c_i \in \mathcal{C}\)}
    \State \(s(q, c_i) \gets f_{\text{dense}}(\mathbf{q}^{\text{dense}}, \mathbf{h}_i^{(K)})\)
\EndFor

\LineComment{Select top-relevant chunks and construct subgraphs.}
\State \(S(\bar{\mathcal{G}}; q) \gets \text{Top-}N\) contextual subgraphs based on \(s(q, c_i)\)

\LineComment{Return the retrieved subgraphs.}
\State\Return{$S(\bar{\mathcal{G}}; q)$}

\end{algorithmic}
\end{algorithm}

The entangled representation of \(c_i\) is obtained through a graph encoder, such as a GNN, which aggregates information based on the contextual graph:
\begin{equation}
    \mathbf{H}_i = g_{\Phi}(\{\mathbf{h}_j^{(0)}, \delta_{qj}, \alpha_{ij} \mid j \in \{i\} \cup \mathcal{N}(i)\}),
\end{equation}
where \(\Phi\) denotes the parameters of the graph encoder, and \(\mathbf{h}_j^{(0)} = \mathbf{c}_j^{\text{dense}}\) is the initial dense representation of \(c_j\). The terms \(\delta_{qj}\) and \(\alpha_{ij}\) control the message passing, ensuring that only relevant information is propagated. Specifically, the message contribution from a neighboring chunk \(c_j\) to \(c_i\) and the subsequent update are defined as:
\begin{align}
    &\mathbf{m}_{j}^{(k)} = \text{MSG}^{(k)}\left(\delta_{qj} \cdot \alpha_{ij} \cdot \mathbf{h}_j^{(k)}\right), \\
    &\mathbf{h}_i^{(k+1)} = \text{AGG}^{(k)}\left(\{\mathbf{m}_j^{(k)} \mid j \in \{i\} \cup \mathcal{N}(i)\}\right),
\end{align}
where \(\mathbf{h}_j^{(k)}\) is the representation of \(c_j\) at layer \(k\).  
The term \(\delta_{qj}\) evaluates the relevance between the query and the context, while \(\alpha_{ij}\) measures the relevance between the central chunk \(c_i\) and its neighboring chunks \(c_j\).
Specifically, \(\delta_{qi}\), the sparse relevance between a query \(q\) and a chunk \(c_i\), is computed as:
\begin{equation}
\delta_{qi} = f_{\text{sparse}}(\mathbf{q}_i^{\text{sparse}}, \mathbf{c}_i^{\text{sparse}}).
\end{equation}
where $f_{\text{sparse}}$ indicates a relevance scoring function used in sparse retrieval such as cosine similarity and dot product. This incorporates sparse relevance into the dense embedding during message passing. To model the dense interaction between \(c_i\) and its neighboring chunks \(c_j \in \mathcal{N}(i)\), \(\alpha_{ij}\) is calculated as:
\begin{equation}
\alpha_{ij} = \text{MLP}_{\phi}(\mathbf{c}_i^{\text{dense}} \ominus \mathbf{c}_j^{\text{dense}}),
\end{equation}
where $\alpha_{ii}$ is defined as 1,  \(\ominus\) represents element-wise subtraction to compute the feature difference, and \(\text{MLP}_{\phi}\) parameterized by \(\phi\), adaptively assesses the relevance between chunks.

This entangled representation framework integrates sparse and dense features while leveraging structural information from the graph context, ensuring that each chunk's representation reflects both its intrinsic relevance and its contextual relationships.

Given a query \(q\) over the chunk set \(\mathcal{C}\) of a citation graph \(\mathcal{G}\), the relevance scoring function \(f_{\text{entangled}}(q, c_i): \mathcal{Q} \times \mathcal{C} \to \mathbb{R}\) is:
\begin{equation}\label{eq:score}
    s(q, c_i) = f_{\text{dense}}(\mathbf{q}^{\text{dense}}, \mathbf{H}_i),
\end{equation}
where $f_{\text{dense}}$ indicates a relevance scoring function used in dense retrieval, \(\mathbf{q}^{\text{dense}}\) represents the dense vector of the query, and \(\mathbf{H}_i\) denotes the entangled representation of the chunk \(c_i\). The overall algorithm is depicted in \hyperref[alg:retrieval]{Algorithm} \ref{alg:retrieval}. 

\begin{proposition}[LeSeGR Generality]
Post-retrieval methods with the metric \(f_{\text{hybrid}}: f_{\text{sparse}} \bigoplus f_{\text{dense}} \to \mathbb{R}\), represent a special case of the proposed Lexical-Semantic Graph Retrieval (LeSeGR). This holds when no additional relevant contextual information exists for any chunk in the citation graph, reducing LeSeGR to existing hybrid retrieval used post-retrieval fusion.
\end{proposition}

\begin{proof}
The entangled representation of a chunk \(c_i\) is given by:
\begin{equation}
    \mathbf{H}_i = \text{AGG}(\{\mathbf{h}_j^{(k)} \mid j \in \{i\} \cup \mathcal{N}(i)\}) = \text{AGG}(\lambda_i \mathbf{m}_i^{(k)}, \lambda_{\mathcal{N}} \mathbf{m}_{\mathcal{N}}^{(k)}),
\end{equation}
where \(\mathbf{m}_i^{(k)}\) and \(\mathbf{m}_{\mathcal{N}}^{(k)}\) denote messages from the central chunk \(c_i\) and its neighboring chunks at layer $k$, respectively. The weights \(\lambda_i\) and \(\lambda_{\mathcal{N}}\) are typically equal and defined as \(\frac{1}{|\mathcal{N}(i)| + 1}\).

When the AGG function is defined as either mean or sum, both of which satisfy the distributive law, the relevance score \(s(q, c_i)\), as defined in \hyperref[eq:score]{Equation} \ref{eq:score}, can be expanded as:
\begin{align}
    s(q, c_i) &= f_{\text{dense}}(\mathbf{q}^{\text{dense}}, \text{AGG}(\lambda_i \mathbf{m}_i^{(k)},\lambda_{\mathcal{N}} \mathbf{m}_{\mathcal{N}}^{(k)})) \\
    &\propto \lambda_i f_{\text{dense}}(\mathbf{q}^{\text{dense}}, \mathbf{m}_i^{(k)}) + \lambda_{\mathcal{N}}f_{\text{dense}} (\mathbf{q}^{\text{dense}}, \mathbf{m}_{\mathcal{N}}^{(k)}).
\end{align}
When no relevant neighbors exist (i.e., \(\mathcal{N}(i) = \emptyset\)), the second term vanishes, leaving:
\begin{align}
    s(q, c_i) &\propto \lambda_i f_{\text{dense}}(\mathbf{q}^{\text{dense}}, \mathbf{m}_i^{(k)}) = \lambda_i \delta_{qi} f_{\text{dense}} (\mathbf{q}^{\text{dense}}, \mathbf{h}_i^{(k)}).
\end{align}
Substituting \(\mathbf{h}_i^{(k)} = \mathbf{c}_i^{\text{dense}}\), the relevance score simplifies to:
\begin{align}
    s(q, c_i) &\propto \delta_{qi} f_{\text{dense}}(\mathbf{q}^{\text{dense}}, \mathbf{c}_i^{\text{dense}}).
\end{align}
Taking the logarithm for further analysis:
\begin{align}
    \log(s(q, c_i)) &\propto \log(\delta_{qi}) + \log(f_{\text{dense}}(\mathbf{q}^{\text{dense}}, \mathbf{c}_i^{\text{dense}})).
\end{align}
Since \(\delta_{qi} = f_{\text{sparse}}(\mathbf{q}^{\text{sparse}}, \mathbf{c}_i^{\text{sparse}})\), this becomes:
\begin{align}
    \log(s(q, c_i)) &\propto \log(f_{\text{sparse}}(\mathbf{q}^{\text{sparse}}, \mathbf{c}_i^{\text{sparse}})) \\&\quad+ \log(f_{\text{dense}}(\mathbf{q}^{\text{dense}}, \mathbf{c}_i^{\text{dense}})),
\end{align}
where the use of \(\log(\cdot)\) requires \( f_{\text{sparse}} \) and \( f_{\text{dense}} \) to be mapped to \( \mathbb{R}^+ \). This requirement can be fulfilled by applying appropriate activation functions or transformations to transform them from \( \mathbb{R} \) into the non-negative domain. In this case, the relevance score is equivalent to the additive fusion of sparse and dense relevance, as used in post-retrieval hybrid methods. Hence, when no contextual information exists (\(\mathcal{N}(i) = \emptyset\)), our graph-contextualized retrieval reduces to the post-retrieval fusion paradigm, demonstrating that the latter is a specific instance of the former.
\end{proof}

When relevant graph contexts are present, the entangled framework dynamically propagates and aggregates sparse and dense signals through structural relationships among neighboring chunks. This enables the model to capture relational dependencies and multi-hop connections in graphs, enhancing retrieval accuracy and effectively utilizing sparse and dense signals from neighbors.

\subsection{Contextualized Graph Retrieval-Augmented Generation (CG-RAG)}\label{sec:generation}
\begin{figure}[tb]
  \centering
  \includegraphics[width=0.5\textwidth]{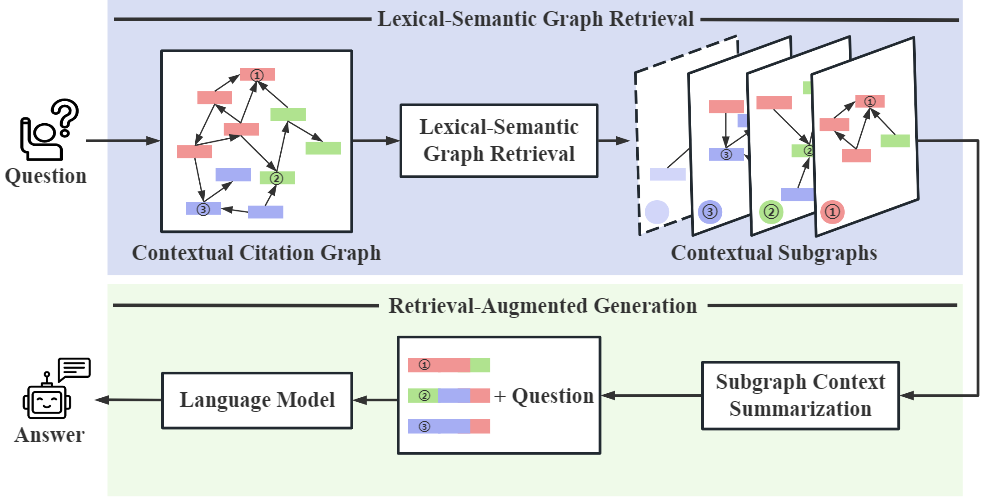}
  \caption{Overview of Contextualized Graph Retrieval-Augmented Generation.}
  \label{fig:rag}
  \vspace{-5pt}
\end{figure}

Given a query $q$ on a citation graph $\mathcal{G}$, we use LeSeGR to retrieve the top \(N\) chunks that are most relevant to the question. These retrieved chunks, together with their contextual subgraph, are then used to generate the answer, as illustrated in \hyperref[fig:rag]{Figure} \ref{fig:rag}. Specifically, for each selected chunk, its corresponding contextual subgraph \(\bar{\mathcal{G}}_i\) is retained for the generation phase, where \(\bar{\mathcal{G}}_i = \bar{\mathcal{G}}[\{i\} \cup \mathcal{N}(i)]\) represents the induced subgraph consisting of chunk \(c_i\) and its direct neighbors. Formally, the set of contextual subgraphs for the Top-\(N\) chunks is defined as:
\begin{equation}
S(\bar{\mathcal{G}}; q) = \bigcup_{c_i \in \text{Top-}N(q, \mathcal{C})} \bar{\mathcal{G}}_i,
\end{equation}
where \(\text{Top-}N(q, C)\) represents the Top-\(N\) chunks ranked by \(s(q, c_i)\). This ensures the retrieval retains both the most relevant chunks and their graph context for downstream tasks.

To effectively utilize the contextual information of each retrieved chunk and adapt to various LLMs, including open-source models such as LLaMA and closed-source models such as ChatGPT, we first summarize the graph context and then concatenate this summarized context with the central chunks to enhance generation.
Specifically, for each contextual subgraph \(\mathcal{\bar{G}}_i \in S(\bar{\mathcal{G}}; q)\), we prompt the LLM to summarize the contextual information surrounding \(c_i\). The summarized context is concatenated with the query to form the final input for generating the answer. The generation process over the citation graph \(\mathcal{G}\) is formally defined as:
\begin{equation}
    p_{\theta}(Y|q, \mathcal{G}) = \arg\max_{\theta} \prod_{i=1}^{n} p_{\theta}(y_i|y_{<i}, X_q, X_{\mathcal{C}}),
    \label{output}
\end{equation}
where \(\theta\) represents the LLM parameters, \(X_q = \text{TextEmbedder}(q)\) is the query embedding, and \(X_{\mathcal{C}}\) is the context embedding:
\begin{equation}
    X_{\mathcal{C}} = \text{TextEmbedder}\left(\left[\text{Summarize}(q, \bar{\mathcal{G}}_i)\right]_{\bar{\mathcal{G}}_i \in S(\bar{\mathcal{G}}; q)}\right),
\end{equation}
representing the embeddings of the concatenated summarized contexts from the retrieved contextual subgraphs. The summarization is performed by the LLM itself, extracting relevant information from the contextual subgraph to aid in generating the final answer.

\section{Experiment}\label{sec:exp}
We conduct experiments to evaluate the effectiveness (\hyperref[exp:main]{Section} \ref{exp:main}) and efficiency (\hyperref[exp:efficiency]{Section} \ref{exp:efficiency}) of Contextualized Graph RAG, along with an analysis of the individual contributions of our technical designs (\hyperref[exp:ablation]{Section} \ref{exp:ablation}). 

\subsection{Settings}
\noindent \textbf{Datasets.}\quad  Our experiments utilize two datasets: PubMedQA-1k and PapersWithCodeQA. PubMedQA-1k is a publicly available dataset, introduced by \citeauthor{jin2019pubmedqa}, and comprises 1,000 question-answer pairs designed for PubMed literature \footnote{\url{https://pubmed.ncbi.nlm.nih.gov/}}, with human-labeled gold-standard retrieval and answer annotations. The original dataset, however, lacks citation information between papers, which we addressed by extracting the references for each paper and constructing a citation graph database with a total of 7,849 papers.

The PapersWithCodeQA dataset was collected from the PapersWithCode website\footnote{\url{https://paperswithcode.com/}}, which tracks research papers across various computer science fields. We used 84 leaderboards\footnote{\url{https://paperswithcode.com/sota}} spanning diverse domains, including \textit{Computer Vision}, \textit{Natural Language Processing}, \textit{Medical}, and \textit{Graphs}. For each leaderboard, we extracted the top 20 papers' contents and references from arXiv\footnote{\url{https://arxiv.org/}} to construct a graph database. The LaTeX content of each paper was preserved as its textual attributes. The dataset comprises 12,171 papers, from which we also crafted 924 questions centered on leaderboard analysis, with ground truth answers derived directly from the leaderboards. These include 420 True/False questions, 420 multiple-choice questions, and 84 generative questions. For generative questions, we first provide the LLMs with the most relevant contexts labeled by humans and allow the LLMs to generate answers. We then evaluate the quality of retrieval-augmented generation by replacing the human-selected contexts with those retrieved by different retrieval methods.

\begin{table}[!tb]
    \centering
    \small
    \caption{Example question-answering pairs and corresponding evaluation metrics.}
    \vspace{-10pt}
    \label{tab:dataset}
    \begin{threeparttable}
    \begin{tabular}{l p{0.32\textwidth}}
    \toprule
        \textbf{True or False: } & Do mitochondria play a role in remodelling lace plant leaves during programmed cell death? Yes, No or Maybe.\\
        \hdashline
        \textbf{Answer: } & Yes\\
        \hdashline
        \textbf{Metrics: } & Accuracy (Acc), $F_1$ Score.\\
        \midrule
        \midrule
        \textbf{Multiple Choice: } & Which model achieves state-of-the-art$^*$ performance on the ADE20K dataset for semantic segmentation? (a) BEiT-3 (b) DINOv2 (c) ONE-PEACE (d) EVA\\
        \hdashline
        \textbf{Answer: } & (c)\\
        \hdashline
        \textbf{Metrics: } & Mean Reciprocal Rank (MRR), Hit@k.\\
        \midrule
        \midrule
        \textbf{Essay: } & Could you provide an overview of the model development for semantic segmentation? \\
        \hdashline
        \textbf{Answer: } & ... Early models like FCN (Fully Convolutional Networks) laid the foundation by adapting classification networks for pixel-level predictions ... Recently, ONE-PEACE emerged as a state-of-the-art model ...\\
        \hdashline
        \textbf{Metrics: } & Coherence, Consistency, Relevance\\
    \bottomrule
    \end{tabular}
    \begin{tablenotes}
        \footnotesize
        \item $^*$Within the citation graph we collected.
    \end{tablenotes}
    \end{threeparttable}
\vspace{-15pt}
\end{table}

\vspace{2pt}
\noindent \textbf{Metrics.}\quad To comprehensively evaluate the performance of RAG systems in retrieving relevant information and generating accurate, contextually appropriate answers, we employ distinct metrics for retrieval and question answering. For retrieval, we use Hit@1 and Hit@3, which measure the proportion of queries where the correct chunk is ranked within the top-1 and top-3 retrieved results, respectively.

For research question answering, example questions and used evaluation metrics are presented in \hyperref[tab:dataset]{Table} \ref{tab:dataset}.
For multiple-choice questions, we use Mean Reciprocal Rank (MRR) and Hit@k to assess ranking quality. For True/False questions, Accuracy (Acc) and $F_1$ score evaluate classification performance. For generative tasks, we leverage the UniEval model \citep{zhong2022towards} to assess Coherence, Consistency, and Relevance, which evaluate logical flow, factual accuracy, and topical alignment, respectively. All metrics adhere to the principle that higher values indicate better model performance.

\begin{table*}[t]
  \vspace{-5pt}
  \centering
  \caption{Evaluation of the retrieval-augmented research question answering. The best performance is highlighted in \textbf{BOLD}, while the second-best performance is \underline{underlined}. Performance of our methods is \colorbox{majorelleblue!25}{highlighted}.}
  \vspace{-10pt}
  \label{tab:main}
  \begin{tabular}{p{0.11\textwidth}p{0.12\textwidth}cc:cc:ccc:cc}
    \toprule
    & & \multicolumn{7}{c}{\textbf{PapersWithCodeQA}} & \multicolumn{2}{c}{\textbf{PubMedQA}} \\
    \cmidrule(lr){3-9}
    \cmidrule(lr){10-11}
    \textbf{Category} & \textbf{Method} &  \textbf{Acc} & \textbf{$F_1$} & \textbf{MRR} & \textbf{Hit@1} & \textbf{Coherence} & \textbf{Consistency} & \textbf{Relevance} & \textbf{Acc} & \textbf{$F_1$} \\
    \midrule
    \multirow{3}{*}{Sparse} 
    & BM25 & 0.689 & 0.617 & 0.765 & 0.736 & 0.905 & 0.858 & 0.859 & 0.662 & 0.604 \\
    & Doc2Query & 0.705 & 0.629 & 0.748 & 0.731 & 0.914 & 0.833 & 0.852 & 0.684 & 0.614 \\
    & BGE-M3 & 0.751 & 0.648 & 0.810 & \underline{0.787} & \underline{0.934} & 0.876 & 0.863 & 0.722 & \underline{0.644} \\
    \midrule
    \multirow{5}{*}{Dense} 
    & MiniLM & 0.730 & 0.644 & 0.782 & 0.758 & 0.919 & 0.872 & 0.828 & 0.712 & 0.641 \\
    & LaBSE & 0.591 & 0.552 & 0.677 & 0.643 & 0.875 & 0.545 & 0.616 & 0.403 & 0.396 \\
    & mContriever & 0.523 & 0.531 & 0.647 & 0.613 & 0.863 & 0.469 & 0.438 & 0.288 & 0.271 \\
    & E5 & 0.579 & 0.560 & 0.659 & 0.628 & 0.872 & 0.521 & 0.544 & 0.363 & 0.360 \\
    & SPAR & 0.611 & 0.583 & 0.643 & 0.609 & 0.886 & 0.549 & 0.537 & 0.392 & 0.384 \\
    \midrule
    \multirow{4}{*}{Hybrid} 
    & Score Fusion & 0.739 & 0.656 & 0.774 & 0.749 & 0.908 & \underline{0.891} & \underline{0.887} & 0.674 & 0.613 \\
    & ColBERT & \underline{0.769} & \underline{0.661} & \underline{0.827} & 0.778 & 0.927 & 0.884 & 0.874 & \underline{0.724} & 0.642 \\
    & CLEAR & 0.618 & 0.575 & 0.667 & 0.643 & 0.894 & 0.623 & 0.685 & 0.468 & 0.456 \\
    \cmidrule{2-11}
    & \textbf{LeSeGR (Ours)} & \cellcolor{majorelleblue!25}{\textbf{0.835}} & \cellcolor{majorelleblue!25}{\textbf{0.703}} & \cellcolor{majorelleblue!25}{\textbf{0.884}} & \cellcolor{majorelleblue!25}{\textbf{0.827}} & \cellcolor{majorelleblue!25}{\textbf{0.956}} & \cellcolor{majorelleblue!25}{\textbf{0.921}} & \cellcolor{majorelleblue!25}{\textbf{0.914}} & \cellcolor{majorelleblue!25}{\textbf{0.778}} & \cellcolor{majorelleblue!25}{\textbf{0.685}} \\
    \bottomrule
  \end{tabular}
\end{table*}

\vspace{2pt}
\noindent \textbf{Implementations.}\quad Experiments are conducted using two NVIDIA A10 GPUs, with Graph Transformer \citep{shi2020masked} serving as the graph encoder. The configuration includes two layers, each featuring four attention heads and a hidden dimension size of 1024. The maximum chunk length is set to 8,192 tokens. Training is conducted using CrossEntropy, with 10\% of the samples labeled with gold-standard retrieval, and optimized using the AdamW optimizer \citep{loshchilov2017decoupled}. The relevance scoring function for both dense and sparse representations is dot product. For generation, GPT-4 is employed through the OpenAI API, specifically leveraging the \texttt{gpt-4o-2024-05-13} model version.

\vspace{2pt}
\noindent \textbf{Baseline methods.}\quad To evaluate the proposed graph-contextualized retrieval method, we benchmark it against several state-of-the-art baselines renowned for their effectiveness in the retrieval phase across various RAG systems. These retrieval techniques are categorized into three groups: sparse retrieval, dense retrieval, and hybrid retrieval.
\begin{itemize}[left=0pt]
    \item \textbf{Sparse Retrieval}: 
    \textit{BM25} \citep{robertson2009probabilistic}, an advanced refinement of TF-IDF \citep{salton1975vector}, enhances relevance scoring by incorporating probabilistic modeling, term saturation, and document length normalization, providing robust performance in keyword-based retrieval tasks. \textit{Doc2Query} \citep{nogueira2019document} improves sparse retrieval by generating synthetic queries for documents using a pre-trained language model (PLM). \textit{BGE-M3} \citep{chen2024bge} utilizes a multi-vector architecture to create robust representations, integrating contrastive learning and knowledge distillation techniques.
    \item \textbf{Dense Retrieval}: \textit{MiniLM} \citep{wang2020minilm} is a lightweight transformer model that uses deep self-attention distillation to create dense embeddings. \textit{LaBSE} \citep{feng2020language} is a bilingual embedding model, leveraging dual encoders and a large-scale parallel corpus to ensure semantic alignment across languages. \textit{mContriever} \citep{izacard2021unsupervised} employs unsupervised contrastive learning to train dense retrievers, focusing on encoding diverse and nuanced contextual information. \textit{E5} \citep{wang2022text} optimizes embeddings for text retrieval by integrating explicit supervision from retrieval datasets and task-specific fine-tuning. \textit{SPAR} \citep{chen2021salient} employs salient phrase representation learning to bridge dense and sparse retrieval, utilizing a dual encoder architecture that explicitly models both phrase-level and document-level semantics.
    \item \textbf{Hybrid Retrieval}: \textit{ScoreFusion} \citep{kuzi2020leveraging} combines the output scores of sparse and dense retrieval models to produce a unified ranking. \textit{ColBERT} \citep{khattab2020colbert} introduces late interaction to compute pairwise term similarities between query and document embeddings, enabling efficient and fine-grained integration of sparse and dense signals. \textit{CLEAR} \citep{gao2021complement} employs a residual learning framework to combine sparse and dense representations, ensuring complementary signals are utilized for improved retrieval performance.
\end{itemize}

\subsection{Main Results} \label{exp:main}
Our proposed Contextualized Graph Retrieval-Augmented Generation with LeSeGR achieves state-of-the-art performance across all tasks and datasets, as shown in \hyperref[tab:main]{Table} \ref{tab:main}. For true/false questions, LeSeGR significantly surpasses sparse, dense, and hybrid baselines in both accuracy (Acc) and $F_1$ scores, demonstrating its capability to effectively capture domain-specific terms and semantic nuances. Hybrid methods such as ScoreFusion combine sparse and dense signals but fail to achieve the deeper integration of retrieval signals offered by LeSeGR. Through its graph-structured integration, LeSeGR dynamically propagates and entangles retrieval signals from contextual information, fully leveraging the relationships embedded in the graph structure. This advanced integration translates to superior performance, particularly in metrics such as MRR and Hit@1.

\begin{table}[h!]
  \vspace{-5pt}
  \centering
  \caption{Evaluation of the retrieval effectiveness on the citation graph of PubMed (PubMedQA). The best performance is highlighted in \textbf{BOLD}, while the second-best performance is \underline{underlined}.}
  \vspace{-10pt}
  \label{tab:retrieval}
  \begin{tabular}{p{0.11\textwidth}p{0.12\textwidth}cc}
    \toprule
    \textbf{Category} & \textbf{Method} &  \textbf{Hit@1} & \textbf{Hit@3} \\
    \midrule
    \multirow{3}{*}{Sparse} & BM25 & 0.835 & 0.912 \\
    & Doc2Query & 0.832 & 0.930 \\
    & BGE-M3 & \underline{0.915} & 0.960 \\
    \midrule
    \multirow{5}{*}{Dense} 
    & MiniLM & 0.887 & 0.945 \\
    & LaBSE & 0.305 & 0.471 \\
    & mContriever & 0.472 & 0.496 \\
    & E5 & 0.231 & 0.355 \\
    & SPAR & 0.256 & 0.385 \\
    \midrule
    \multirow{4}{*}{Hybrid} 
    & Score Fusion & 0.829 & 0.925 \\
    & ColBERT & 0.913 & \underline{0.968} \\
    & CLEAR & 0.470 & 0.612 \\
    \cmidrule{2-4}
    & \textbf{LeSeGR} (Ours) & \cellcolor{majorelleblue!25}{\textbf{0.961}} & \cellcolor{majorelleblue!25}{\textbf{0.987}} \\
    \bottomrule
  \end{tabular}
  \vspace{-10pt}
\end{table}

In generative tasks, our method demonstrates significant improvement in Coherence, Consistency, and Relevance by leveraging its entangled sparse-dense representation and graph-based contextualization. Unlike hybrid baselines such as ColBERT, which emphasizes token-level interactions but overlooks graph-level relationships, LeSeGR's graph encoder dynamically aggregates signals across interconnected chunks. This enhances contextual understanding, enabling high-quality and contextually rich text generation. For instance, LeSeGR achieves a Coherence score of 0.956 on PapersWithCodeQA, outperforming ColBERT's 0.927. These results underscore the unique strengths of LeSeGR in effectively bridging sparse and dense retrieval with graph-based contextualization, thereby advancing the retrieval-augmented generation process to new levels of effectiveness.

\hyperref[tab:retrieval]{Table} \ref{tab:retrieval} demonstrates that LeSeGR significantly outperforms all baselines in the retrieval phase, including sparse, dense, and hybrid approaches. Our method achieves superior retrieval accuracy by effectively entangling sparse and dense signals within the graph structure, allowing contextual information to enhance relevance scoring. Unlike post-retrieval fusion methods, such as ScoreFusion, which combine sparse and dense signals after separate retrieval processes, our approach dynamically integrates these signals during retrieval, leading to more coherent and effective utilization of both lexical and semantic information. Furthermore, while ColBERT performs token-level interactions for fine-grained relevance, it operates at the query-document level without fully leveraging the structural relationships present in citation graphs. In contrast, our method extends relevance computation to the graph structure, propagating and aggregating signals across related chunks to capture multi-hop and relational dependencies. This deeper integration of graph context and entangled sparse-dense signals enables our method to outperform ColBERT, achieving the highest Hit@1 and Hit@3 scores.

\subsection{Efficiency Analysis} \label{exp:efficiency}

As shown in \hyperref[tab:efficiency]{Table} \ref{tab:efficiency}, LeSeGR demonstrates competitive retrieval efficiency on the citation graph of arXiv (PapersWithCodeQA). It strikes a balance between memory usage and latency, leveraging GPU computation effectively. LeSeGR achieves faster query latency (403.94 ms) compared to ColBERT (561.91 ms) while maintaining a moderate GPU memory footprint (1,921 MB). In contrast, ScoreFusion exhibits high CPU memory usage (5,655 MB) and slower query speeds, whereas LeSeGR optimizes GPU utilization by integrating both sparse and dense retrieval signals into the message passing process of the graph encoder. Additionally, LeSeGR outperforms CLEAR in query speed while maintaining similar memory usage. These results highlight LeSeGR's efficiency and scalability for large-scale graph-based retrieval tasks without compromising effectiveness.

\begin{table}[!h]\small
  \vspace{-5pt}
  \centering
  \caption{Evaluation of the retrieval efficiency on the citation graph of arXiv (PapersWithCodeQA).}
  \label{tab:efficiency}
  \vspace{-10pt}
  \begin{tabular}{p{0.09\textwidth}cccc}
  \toprule
  \multirow{2}{*}{Method} & \multicolumn{2}{c}{\textbf{CPU \& GPU Memory}} & \multicolumn{2}{c}{\textbf{Indexing \& Query Latency}}\\
  \cmidrule(lr){2-3} \cmidrule(lr){4-5} 
  & (MB) & (MB) & (ms) & (ms) \\
  \midrule
  Score Fusion & 5,655 & 770 & 43.94 & 1,580.14 \\
  ColBERT & 0 & 12,674 & 12.40 & 561.91\\
  CLEAR & 0 & 1,538 & 205.36 & 16.07\\ 
  \midrule
  \textbf{LeSeGR} & 0 & 1,921 & 19.22 & 403.94\\
  \bottomrule
  \end{tabular}
\end{table}

\subsection{Ablation Studies} \label{exp:ablation}

The ablation studies performed for each influential factor in LeSeGR is shown in \hyperref[tab:ablation]{Table} \ref{tab:ablation}. In this experiment, we evaluate LeSeGR on the citation graph of PubMed, which contains 7,849 papers. Our main observations are as follows:

\begin{table}[h]
    \centering
    \caption{Ablation studies on PubMedQA. The default settings of LeSeGR are marked with \textit{*}.}
    \label{tab:ablation}
    \vspace{-10pt}
    \begin{tabular}{p{1.8cm}lccc}
        \toprule
        \textbf{Factor} & \textbf{Setting} & \textbf{Hit@1} & \textbf{Hit@3} \\
        \midrule
        \multirow{3}{2cm}{Graph\\Encoder}
        & GAT & 0.939 & 0.968 \\
        & GCN & 0.955 & 0.976 \\
        & {Graph Transformer}$^*$ & \textbf{0.961} & \textbf{0.987}  \\
        \midrule
        \multirow{3}{2cm}{Top-$n$\\Context} 
        & 2 & 0.931 & 0.949 \\
        & 8 & 0.950 & 0.984 \\
        & {4}$^*$ & \textbf{0.961} & \textbf{0.987}  \\
        \midrule
        \multirow{4}{2cm}{Sparse\\Signal} 
        & TF-IDF & 0.903 & 0.962 \\
        & BM25 & 0.919 & 0.962 \\
        & Doc2Query & 0.926 & 0.964 \\
        & {BGE-M3}$^*$ & \textbf{0.961} & \textbf{0.987}  \\ 
        \midrule 
        \multirow{3}{2cm}{Dense\\Signal}
        & E5 & 0.676 & 0.765 \\
        & mContriever & 0.838 & 0.883 \\
        & {MiniLM}$^*$ & \textbf{0.961} & \textbf{0.987} \\
        \bottomrule
    \end{tabular}
\end{table}

\begin{itemize}[left=0pt]
    \item \textbf{Graph Encoder.}\quad Among Graph Attention Networks (GAT) \citep{velivckovic2017graph}, Graph Convolutional Networks (GCN) \citep{kipf2016semi}, and Graph Transformer \citep{shi2020masked}, Graph Transformer achieves the highest Hit@1 and Hit@3 scores of 0.961 and 0.987, respectively. Notably, regardless of which graph encoder is employed, our LeSeGR method still consistently achieves the best retrieval performance when compared to the baselines in \hyperref[tab:efficiency]{Table} \ref{tab:retrieval}. This underscores LeSeGR's superior ability to model complex relationships and effectively aggregate contextual information for retrieval.
    \item \textbf{Top-\(n\) Context.}\quad We further assess three configurations for the number of contexts connected via inter-document edges, i.e., Top-\(n\). The \(n = 4\) setting achieves the best results, striking a balance between sufficient contextual inclusion and noise reduction. Smaller values, such as \(n = 2\), restrict the scope of context, while larger values, such as \(n = 8\), may introduce irrelevant information, diluting the positive impact of relevant context and reducing retrieval effectiveness. It is anticipated that if the chunk length is sufficiently large, retaining only the top-1 relevant chunk will suffice.
    \item \textbf{Retrieval Signals.}\quad We compare four sparse representation methods. BGE-M3 outperforms other sparse encoders, achieving the highest scores, as its ability to integrate lexical and semantic features is critical for domain-specific term matching. TF-IDF and BM25, while strong in lexical precision, lack semantic adaptability. In addition, we also compare three dense representation methods, with MiniLM delivering the best performance. MiniLM's compact representation effectively captures semantic nuances, making it better suited for diverse queries and documents. As a whole, however, combining two expressive retrieval models with our LeSeGR framework results in stronger retrieval performance. Notably, the performance of LeSeGR appears to be primarily constrained by the quality of the dense retrieval signal.
\end{itemize}

\section{Conclusion}\label{sec:conclution}

In this work, we introduce Lexical-Semantic Graph Retrieval (LeSeGR), a novel framework that integrates sparse, dense, and graph-structured retrieval signals for complex and structured database. Based on LeSeGR, we present Contextualized Graph Retrieval-Augmented Generation (CG-RAG) for research question answering. By leveraging a contextual citation graph, our approach effectively captures intra- and inter-document relationships, enabling a dynamic propagation of contextual information through an entangled hybrid retrieval paradigm. This paradigm bridges lexical precision and semantic understanding while generalizing to existing retrieval methods. Furthermore, CG-RAG incorporates a graph-aware generation strategy, enhancing the contextual richness of generated responses. Extensive experiments across multiple citation networks demonstrate the superior performance of CG-RAG based on LeSeGR, achieving state-of-the-art results in retrieval metrics such as Hit@1 and generation metrics such as Coherence and Relevance. Our findings underscore the effectiveness of graph-contextualized representations in advancing the capabilities of retrieval-augmented generation for citation graphs, setting a new benchmark for retrieval-augmented research question answering.


\bibliographystyle{ACM-Reference-Format}
\bibliography{sample-base}



\end{document}